\documentclass{article}

\usepackage[utf8]{inputenc}
\usepackage{amsmath,amsthm,latexsym,amssymb,amsfonts,color}
\usepackage{bbm}

\usepackage{hyperref}
\usepackage{orcidlink}

\newcommand{\R}{{\mathbb R}}
\newcommand{\C}{{\mathbb C}}

\newcommand{\Lie}{{\mathcal L}}

\DeclareMathOperator{\const}{const}
\DeclareMathOperator{\rd}{d}

\newtheorem{df}{Definition}
\newtheorem{lm}{Lemma}
\newtheorem{cor}{Corollary}
\newtheorem{thm}{Theorem}
\newtheorem{prop}{Proposition}

\theoremstyle{remark}
\newtheorem{remark}{Remark}

\usepackage{graphicx} 

\title{Extreme horizon equation}
\author{Wojciech Kamiński \orcidlink{0000-0003-3707-6087}, Jerzy Lewandowski
 \orcidlink{0000-0001-8512-1490}}
\date{%
\small
  {\it   Faculty of Physics, University of Warsaw,\\
ul. Pasteura 5, 02-093 Warsaw, Poland}\\[2ex]%
    \today
}

\begin{document}

\maketitle

\begin{abstract}
Extremal horizons satisfy an equation induced by the Einstein vacuum equations  that determines the shape of the horizon and the manner in which it rotates (the EEH equation). Until recently, however, the classification of solutions required the assumption of axial symmetry. Recently, there has been a breakthrough: Dunajski and Lucietti proved that every non-static solution possesses a one-dimensional symmetry group. The first part of our work is inspired by this  result. An identity satisfied by the solutions of the EEH equation has been distilled (Master Identity), which is crucial for studying their properties.   It is a bit stronger than the original Dunajski-Lucietti identity and leads directly to the rigidity theorem for any value of the cosmological constant.   Master Identity is used for  a simple derivation of the local form of the general static  solution of the EEH equation with non-positive cosmological constant. All the globally defined compact static  solutions are derived.  Thus  the list of solutions given in the literature is completed.  In the two-dimensional case (which corresponds to horizons in four-dimensional spacetime), the Einstein-Maxwell  equations  of an extremal horizon (EMEH)  and the equations of quasi-Einstein spaces  are studied. The general solution on a compact surface with non-zero genus is derived.  In the case of zero genus,  the static solutions are investigated  and  their axial symmetry is proven.   Together with the new results on non-static solutions on sphere of Colling, Katona and Lucietti that leads to the uniqueness of the Reissner–Nordström-(Anti)de-Sitter extremal horizons. Interestingly,  the static rigidity result is also valid for non-compact spaces with a zero first cohomology group.
\end{abstract}

\section{Introduction: extreme horizons}
The isolated horizons program involves studying the horizons of black holes, as well as cosmological horizons, in a quasi-local manner abstracting from the rest of space-time \cite{ABF1, ABL1, ABL2, AK}. 
Horizon properties are described by an induced degenerate null surface metric and a $1$-form rotation.   Einstein's equations satisfied by space-time induce equations for the geometry of the horizon \cite{ABL1}. 
The case of the extreme horizon is particularly interesting. Mathematicians call it degenerate, while physicists assign it the temperature equal to zero. The vacuum extremal horizon (EEH) equation imposed on the metric tensor and the rotation $1$-form defined on spacelike slice of the horizon is the subject of this paper.  The systematic study of this equation began as part of the isolated horizon program \cite{ABL1,LPextremal,LPhigher,LPJfol}, the equation also appeared earlier in the literature  \cite{Hajicek,IM}. The results that attracted interest were:   the uniqueness of the extremal Kerr horizon geometry as an axially symmetric solution of the EEH equation with zero cosmological constant on a 2-sphere \cite{Hajicek, LPextremal} and the extremal Kerr-Newman horizon in the presence of Maxwell field (EMEH equation) \cite{LPextremal}, the natural generalization of the EEH equation to any spacetime dimension  \cite{LPhigher}, and the relationship with the so-called Near Horizon Geometries \cite{LPJfol,Kunduri2007,KunduriVac,Kunduri:2008tk,NHG}. The uniqueness of axially symmetric solutions defined on $2$-sphere  were generalized to the presence of the cosmological constant \cite{Kunduri:2008tk,KunduriVac,Buk:2020ttx}, to the Einstein-Yang-Mills theory   \cite{Li,LiLucietti13}. The topology of the $2$ sphere was shown to be distinguished for the EEH equation in $2$-dimensions,  since for surfaces of higher genus the only solutions have zero  rotation 1-form potential and constant curvature \cite{DKLS2}. 

The same property is possessed by all the 2-dimensional static solutions (that is of a closed rotation 1-form potential), including those defined on the 2-sphere in this case \cite{CRT, Kunduri:2008tk}. Similar limitations have been proved for static solutions of the EEH equation  in higher dimensions, in the case of the non-positive cosmological constant, but in this case there do exist static solutions with non-zero rotation 1-form \cite{Bahuaud:2022iao}. 
Such solutions were subsequently classified \cite{Wylie2023}, although some topologically non-trivial cases were overlooked in the literature. We consider this case separately below.

The breakthrough result has been  the intrinsic rigidity theorem \cite{DL-rigidity},  first shown in the case of non-negative cosmological constant and subsequently generalized to arbitrary cosmological constant case \cite{Colling2024} and recently to the case with electromagnetic field in dimension $2$ \cite{CKL-new}.  It states that on a compact manifold, every solution of the  corresponding extremity equation  possesses a one-dimensional symmetry group. The intrinsic rigidity had been conjectured earlier, and it had been proven for linear perturbations of the extremity equation  about the Kerr horizon \cite{JK12,  CST}.  In the two-dimensional case, this result combined with the previous ones completely solves the problem. Indeed,  the general solution is known in exact form, on topological $2$-sphere it corresponds to the extremal Kerr, Kerr-de-Sitter or Kerr-Anti-de-Sitter spacetimes  \cite{LPextremal,Kunduri:2008tk,KunduriVac,Buk:2020ttx}, while those on a higher genus surfaces have constant curvature equal to the cosmological constant and  zero rotation 1-form  \cite{DKLS2}, they correspond to horizons in extremal A-metrics spacetimes (also called topological black holes). 

It should be mentioned, that an equivalent equation to the EEH equation  is satisfied by geometry of a null, non-expanding and shear free foliation \cite{LSW}. If the foliation is emanating from an extremal isolated horizon, then the geometric mechanism is understood, however it is satisfied also in the absence of an extremal horizon.  

The first part of the current work was inspired by the groundbreaking proof of the rigidity of solutions to the EEH equation by Dunajski and Lucietti \cite{DL-rigidity}. First, we presented a small improvement of the crucial Andersson-Mars-Simon  theorem \cite{AMS} on certain operator emerging in the context of marginally trapped surfaces \cite{ABL1}. It is a simple corollary of the proof from \cite{ABL1}, however it was not stated explicitly in the literature.  Although, this improvement is not necessary for the proof of the rigidity theorem, the Andersson-Mars-Simon theorem seems important and the property shown by us might be useful in future development. We also formulated and proved a  necessary and sufficient condition for the existence of a Killing field on a compact Riemann space.   We have extracted a key identity (Master Identity) that is satisfied by the solutions of the EEH equation.  Our identity contains the identity derived in \cite{DL-rigidity}, but it is a bit stronger and leads us directly to a simple improvement of the proof, which is valid for any value of the cosmological constant. The original method \cite{Colling2024} requires separate argument.   

The fourth topic of our work is the application of our Master Identity to the static solutions of the EEH equation of  a non-positive cosmological constant.  After a simple derivation of the local form of the  general solution, we constructed all globally defined general solutions. We thus complete the list of solutions given in the literature \cite{Bahuaud:2022iao, Bahuaud:2023wsi, Wylie2023}. For application of original Dunajski-Lucietti method in this case see \cite{Cochran2024killing}.

In the two-dimensional case (which corresponds to horizons in four-dimensional spacetime), we study a generalization of the EEH equation which includes extremal horizon in the presence of a Maxwell field and the equations of quasi-Einstein 
spaces. We will call this class of equations, generalized extremal horizon equations. We find the general solution for a compact surface with non-zero genus. In the case of zero genus,
we examine the static case and prove the axial symmetry of the solutions. Interestingly, this latter result
is also valid for non-compact spaces with a zero first cohomology group.

In the current paper we present several new results concerning the extremity equation. In particular, we derive new static solutions (closed rotation 1-form potential) with non-vanishing rotation 1-form potential that were overlooked in the literature, and we obtain the general static solution with negative cosmological constant. We also consider a generalization of the the extremity equation in dimension $2$ which contain both EMEH and so called $m$-quasi Einstein equations and find general solution of genus $>0$. Additionally, we revisit the proof of the rigidity theorem and presented a simplification of the original argument. One advantage is that our version includes both values of the cosmological constant in a uniform way. We apply our tools provided by this proof to address the static case and fill in the gaps in the other papers on that subject.

\bigskip

In the remainder of this section, we  present the spacetime origin of the EEH equation and EMEH equation studied in this paper from Einstein's theory of the gravitational field. The mathematical formulas and definitions we invoke will not appear directly in our work, although they gave rise to the problems we study and for them the results found are relevant.  
\bigskip

Extremal isolated horizon is an $n+1$ dimensional  manifold ${\cal H}$ endowed with a symmetric, twice covariant  tensor field $g_{ab}$  of the rank  $n$ (called a degenerate metric tensor),  a covariant, torsion free derivative $D_a$, and a nowhere vanishing  vector field $\ell^a$ defined modulo rescaling by a constant, such that the following equalities are satisfied:
\begin{equation}
D_a g_{bc} = 0, \ \ \ \ell^a g_{ab} =0, \ \ \  [{\cal L}_\ell, D_a]=0,  \ \ \ \ell^aD_a \ell^b = 0.
\end{equation}

It follows from the first two equations that 
\begin{equation}
{\cal L}_\ell g_{ab}=0, \ \ \ \nabla_a\ell^b=\omega_a\ell^b, \ \ \ {\rm and}\ \ \ {\cal L}_\ell \omega_a=0
\end{equation}  
where the existence of $\omega_a$ follows from the fact that $g_{ab}$ has exactly one degenerate direction. The $1$-form  $\omega_a$ is called the rotation $1$-form potential.

    Suppose, that an extremal isolated horizon $({\cal H},g_{\alpha\beta},D_\alpha)$ is embedded in $n+2$ dimensional spacetime $(\tilde{M}, \tilde{g}_{AB})$,  such that  $g_{ab}$ and  $D_a$  are equal to the restriction of the spacetime metric tensor $\tilde{g}_{AB}$ and the corresponding covariant derivative $\tilde{\nabla}_A$, respectively, to ${\cal H}$.  Consider a $n$-dimensional section 
    \begin{equation}
         \Sigma \subset{\cal H} 
    \end{equation}
  transversal to $\ell^a$, and abusing the notation, denote by $g_{\alpha\beta}$ and $\omega_\alpha$, the pullbacks of $g_{ab}$ and $\omega_a$ to $\Sigma$.  In $\Sigma$, $g_{\mu\nu}$ becomes a metric tensor, it determines a torsion free, metric covariant derivative $\nabla_\mu$, the Riemann tensor $R^\mu{}_{\nu\alpha\beta}$ and the Ricci tensor $R_{\mu\nu}$.  On the other hand, the pullback  $\tilde{R}_{\mu\nu}$  of the spacetime Ricci tensor on $\Sigma$ is determined by $g_{\mu\nu}$ and $\omega_\mu$ in the following way
    \begin{equation}\label{id}
        \tilde{R}_{\mu\nu} = R_{\mu\nu} - 2\omega_\mu\omega_\nu - \nabla_\mu\omega_\nu - \nabla_\nu\omega_\mu. 
    \end{equation}
   If the spacetime metric tensor satisfies the vacuum Einstein equations with cosmological constant $\Lambda$,
   namely
   \begin{equation}
       \tilde{R}_{AB} = \frac{2}{n}\Lambda \tilde{g}_{AB},
   \end{equation}
  then on $\Sigma$, the identity (\ref{id}) becomes  the EEH equation of Def. \ref{extreq}  \cite{LPhigher}. 

More generally, if the spacetime $(M,\tilde{g}_{AB})$ that admits an extremal isolated horizon satisfies Einstein's equations with some additional fields, then the left hand side in (\ref{id}) will become
\begin{equation}
       \tilde{R}_{\mu\nu}  = 
       \kappa \left(\tilde{T}_{\mu\nu}-\frac{1}{2}\tilde{T}{g}_{\mu\nu}\right),
\end{equation}
where $\tilde{T}_{AB}$ is the energy-momentum tensor (including the cosmological constant) and $\tilde{T}=\tilde{g}^{AB}\tilde{T}_{AB}$. Also,  some of the fields entering ${T}_{\mu\nu}$ and restricted to $\Sigma$, may be constrained by additional equations along $\Sigma$. This is what happens when we allow Maxwell field on ${\cal H}$ and obtain the EMEH equation of Definition \ref{etremmaxeq} \cite{LPextremal}.

 \subsection{Conventions}
Given a metric tensor $g_{\alpha\beta}$,  we use  its  torsion free and metric covariant derivative $\nabla_\mu$, 
\begin{equation}
\nabla_\alpha g_{\beta\gamma} = 0, \ \ \ \ \ (\nabla_\alpha\nabla_\beta -  \nabla_\beta\nabla_\alpha)f = 0, \ \   
\end{equation}
where the second equation holds for every function $f$. The Riemann tensor $R^\alpha{}_{\beta\gamma\delta}$   and Ricci tensor $R_{\alpha\beta}$ are defined as follows.
\begin{equation}
(\nabla_\mu\nabla_\nu-\nabla_\mu\nabla_\nu)V^\rho=R^\rho_{\phantom{\rho}\chi\mu\nu}V^\chi,\quad R_{\mu\nu}=R^\rho_{\phantom{\rho}\mu\rho\nu}.
\end{equation}
The Lie derivative o a tensor $T^{\alpha...\beta}{}_{\gamma...\delta}$ with respect to a vector field $K$ will be denoted by 
${\cal L}_{K}T^{\alpha...\beta}{}_{\gamma...\delta}$. 

\section{Characters of this paper.} To save the reader's time, we begin our article with a direct definition of the extremal horizon equations considered in the paper  (Definition \ref{extreq}, \ref{etremmaxeq}).  Some of our results are more general,  concern the general properties of the Killing wave operator and the Andersson-Mars-Simon operator. We also provide definitions of these operators in advance in this section (Definition \ref{Killop}-\ref{AMSop}).

\begin{df}\label{extreq} \cite{LPhigher} {\bf The Einstein vacuum extremal horizon (EEH) equation.} Given a $n$-dimensional  manifold $\Sigma$, a metric tensor $g_{\alpha\beta}$, and a differential $1$-form $\omega_{\alpha}$ (rotation potential), the vacuum extremal horizon (EEH) equation with cosmological constant $\Lambda$ reads 
 \begin{equation}\label{theequation}
       \nabla_{(\nu}\omega_{\mu)} + \omega_\mu\omega_\nu - \frac{1}{2}R_{\mu\nu} + 
        \frac{1}{n}\Lambda g_{\mu\nu} =0,
   \end{equation}
\end{df}

In Sec. \ref{sec:case2} we will consider a generalization of the EEH equation to the  case with Maxwell field. We will restrict ourselves to  $2$-dimensional horizon sections that correspond to $4$-dimensional spacetimes. 

\begin{df}\label{etremmaxeq} \cite{LPextremal} {\bf The Einstein-Maxwell extremal horizon (EMEH) equation.}  Given a  $2$-dimensional  manifold $\Sigma$, endowed with: a metric tensor $g_{\alpha\beta}$, Hodge dual $*$, a differential $1$-form $\omega_{\alpha}  (rotation potential)$, and complex valued function $\Phi$, the Einstein-Maxwell vacuum extremal horizon (EMEH) equation with cosmological constant $\Lambda$ is the following system 
 \begin{align}
       \nabla_{(\nu}\omega_{\mu)} + \omega_\mu\omega_\nu - \frac{1}{2}R_{\mu\nu} + 
        \frac{1}{2}\Lambda g_{\mu\nu} + |\Phi|^2 g_{\mu\nu} &=0,\\
        (1+i*)\left(d + 2\omega \right)\Phi &= 0.
   \end{align}
\end{df}

Some of the results applied to more general situation.

\begin{df}\label{etremgeq}  {\bf The generalized extremal horizon (GEH) equation.}  Given a  $2$-dimensional  manifold $\Sigma$, endowed with: a metric tensor $g_{\alpha\beta}$,  a differential $1$-form $\omega_{\alpha}$,  the generalized extremal horizon (GEH) equation is the following system 
 \begin{align}
       \nabla_{(\nu}\omega_{\mu)} + a\omega_\mu\omega_\nu +f g_{\mu\nu}=0
   \end{align}
for some smooth function $f\in C(\Sigma)$ on $\Sigma$ and a constant $a\not=0$.
\end{df}

Let us notice that both Einstein vacuum and Einstein-Maxwell extremal horizons belong to this class, as well as so called $m$-quasi Einstein metrics.

Below we will also use the following operator:

\begin{df}\label{Killop} \cite{Wald} {\bf Killing "wave" operator.}
    The Killing "wave" operator $\square$ maps vector fields defined on a manifold $\Sigma$ equipped with a metric tensor $g_{\mu\nu}$ in the following way:
\begin{equation}\label{square}
\square F^\rho:=\nabla^\nu\nabla_\nu F^\rho+R^\rho{}_\mu F^\mu,
\end{equation}
where $R^\rho{}_\mu$ stands for the Ricci tensor of $g_{\mu\nu}$.
\end{df}
It was so named in the context of space-time metric tensors, but the definition does not depend on the signature of the metric, and, as we will see below, a strictly positive signature provides strong constraints.  

\begin{df}\label{AMSop} \cite{ABL1,AMS} {\bf The Andersson-Mars-Simon operator}
Given a manifold $\Sigma$ endowed with a metric tensor $g_{\mu\nu}$ and a differential $1$-form $f_\mu$,  the AMS operator 
is defined for a function $u$ by
\begin{equation}
     H u = -\nabla^\mu \nabla_\mu u+ \nabla^\mu (u f_\mu) . 
\end{equation}

\end{df}

 \section{Generally true geometric results on compact Riemannian spaces.} 
In this section, we separate those results that are generally true and are not yet related to the EEH equation studied later, but they are on the one hand important in their own right, and on the other hand, find significant use in the next section.  
 
\subsection{The Killing wave equation operator on compact Riemannian spaces.} 
 The following identities are true for the Killing "wave" operator (Def. \ref{Killop}):
\begin{equation}
\square F_\mu=2\nabla^\nu \left(\nabla_{(\mu}F_{\nu)}-\frac{1}{2}g_{\mu\nu} \nabla^\chi F_\chi\right).
\end{equation}
That implies that if $F^\mu$ is a Killing vector then the Killing wave equation is satisfied
\begin{equation}
 \square F^\mu = 0.     
\end{equation}
In general the inverse  implication is not true,  even if vector is divergence-free, namely  a divergence free vector field $F^\mu$ that satisfies the Killing equation  may not be a Killing vector field. It turns out, however, that on a compact manifold, if the signature of the metric tensor is strictly positive, then 
\begin{equation}
    \nabla_\mu F^\mu=0,\ \square F^\mu=0\Longleftrightarrow F^\mu \text{ is a Killing vector.}
\end{equation}
We formulate  now and prove a stronger version of this fact. 

\begin{thm}\label{tm:Killing}{\bf Killing vectors and Killing operators.}
Consider a compact manifold $\Sigma$ endowed with a  Riemannian  metric tensor $g_{\mu\nu}$. 
Suppose that a vector field $F^\mu$ defined on $\Sigma$ is divergence free, that is
\begin{equation*}\label{}
    \nabla_\mu F^\mu = 0,
\end{equation*}
and satisfies the following equalities, for some function $U$ and a $1$-form $A_\mu$,
\begin{equation}\label{boxF}
\square F_\nu= \nabla_\nu U+A_\nu,\quad A_\nu F^\nu=0;
\end{equation}
then $F^\nu$ is a Killing vector field and additionally $A_\mu=-\nabla_\mu U$.
\end{thm} 

\begin{remark}
\label{rm:Killing} A consequence of this theorem  is  
\begin{equation}
    \Lie_K U=K^\mu \nabla_\mu U=-K^\mu A_\mu=0.
\end{equation}
It means that the function $U$ is also invariant.
\end{remark}

\begin{proof}
Using $\nabla^\mu F_\mu=0$ we can write
\begin{equation}
2\nabla_{(\mu}F_{\nu)}\nabla^{(\mu}F^{\nu)}=2\nabla_\mu\left(F_\nu \nabla^{(\mu}F^{\nu)}\right)-F^\nu \square F_\nu.
\end{equation}
Moreover, again using $\nabla^\mu F_\mu=0$ and $A_\mu F^\mu=0$
\begin{equation}\label{eq:F-F}
F^\nu \square F_\nu=F^\nu\nabla_\nu U=\nabla_\nu\left(F^\nu U\right),
\end{equation}
so $2\nabla_{(\mu}F_{\nu)}\nabla^{(\mu}F^{\nu)}=\nabla_\nu\left(2F_\mu \nabla^{(\mu}F^{\nu)}-F^\nu U\right)$. Integrating over compact manifold $\Sigma$
\begin{equation}
2\int_\Sigma \nabla_{(\mu}F_{\nu)}\nabla^{(\mu}F^{\nu)}=0\Longrightarrow \nabla_{(\mu}F_{\nu)}=0.
\end{equation}
The identity \eqref{eq:F-F} simplifies as $\square F_\nu=0$ to
\begin{equation}
    0=\nabla_\mu U+A_\mu,
\end{equation}
that is the last statement of the theorem.
\end{proof}

\begin{remark}
It is not difficult to complicate the right-hand side of (\ref{boxF}) even further,  namely instead of assuming $A_\mu K^\mu=0$ we can impose a condition
\begin{equation}\label{boxF'}
A_\mu F^\mu=\nabla_\mu D^\mu+E
\end{equation}
where $E\geq 0$ leaving the conclusion true. It looks a bit baroque, however it works for the EMEH equation \cite{CKL-new}. We should, however stress that in this case our simplification presented in the next section does not work straightforwardly and the original approach as in \cite{CKL-new} seems better.
\end{remark}

\subsection{The Andersson-Mars-Simon operator}
The AMS operators (Def. \ref{AMSop}) feature in the context of horizons or marginally trapped surfaces \cite{AMS, DL-rigidity}. Since they are interesting on their own, let us formulate an upgraded version of the  Andersson-Mars-Simon theorem \cite{AMS}.  

For a smooth Riemannian metric $g_{\mu\nu}$ and smooth $1$-form $f_\mu$ on a compact manifold $\Sigma$ we define an unbounded operator on $L^2(\Sigma)$
\begin{equation}
    Hu:=-\Delta u+\nabla^\mu(f_\mu u)
\end{equation}
with the domain $D(H)=H^2_2(\Sigma)$ (Sobolev space of function whose second weak derivatives are square integrable). It is a closed operator (see \cite{HormanderPDEIII}) with adjoint $H^\dagger$ equal
\begin{equation}
    H^\dagger u=-\Delta u+f_\mu\nabla^\mu u
\end{equation}
with the same domain. The theory of elliptic operators \cite{HormanderPDEIII} ensures that both $H$ and $H^\dagger$ have only discrete spectrum with poles of resolvents of finite orders. This means that for every $\mu\in \C$ the space of generalized eigenfunctions $V_\mu$ defined by
\begin{equation}
    V_\mu=\{\psi\in L^2(\Sigma)\colon \exists_{n>0}\  \psi\in D(H^n),\ (H-\mu)^n\psi=0\}
\end{equation}
is finite dimensional. Moreover, as $H$ is elliptic the spaces $V_\mu$ consists in fact of smooth functions. The same applies to $H^\dagger$.

An eigenvalue $\mu$ is simple if $V_\mu$ consists only of eigenfunctions, namely for every $\psi\in V_\mu$, $H\psi=\mu \psi$. Operator $H$ is not self-adjoint and it is not guaranteed that eigenspaces are simple.

It was noticed in \cite{AMS} and used in \cite{DL-rigidity} that we can tell much more about the spectrum thanks to the Krein-Rutman theory.

\begin{thm}\label{tm:AMS}{\bf Ground state of AMS operators}
Consider a compact manifold $\Sigma$ equipped with a Riemannian metric tensor $g_{\mu\nu}$ and a $1$-form $f_\mu$.  
The operator $H$
\begin{equation}
Hu:=-\Delta u+\nabla^\mu(f_\mu u)
\end{equation} 
has the following properties:
\begin{enumerate}
\item  It has an eigenvalue $0$ with a smooth eigenvector $\Psi$, where  $\Psi>0$ in every point of  $\Sigma$. This eigenvalue is simple and multiplicity free. Every other eigenvalue has real part bigger than $0$.
\item The same holds for $H^\dagger$. The eigenvector for eigenvalue $0$ is a constant function.
\end{enumerate}
If there exists a unitary group $U_t$ preserving $H$ (that is $U_tHU_t^{-1}=H$) and such that $U_t$ preserves the space of non-negative functions then $U_t\Psi=\Psi$.
\end{thm}

\begin{remark}
What is not  explicitly stated in \cite{AMS} (however effectively proved there), is  that the  eigenvalue $0$ is simple (not only multiplicity free) and that in fact every other eigenvalue needs to have a bigger (not equal) than zero  real part.  For that reason we decided to provide a  comment on the proof with suitable amendments in an appendix \ref{sec:AMS-comments}.
\end{remark}

\subsection{Dunajski-Lucietti  vector field.}
\begin{prop}\label{prop:divfreeK}{\bf Dunajski-Lucietti vector.}
Consider a compact manifold $\Sigma$ equipped with a Riemannian metric tensor $g_{\mu\nu}$ and a $1$-form $f_\mu$.  
There exists a nontrivial function $\Gamma$ such that the vector field 
\begin{equation}
K_\mu:=\nabla_\nu \Gamma+\Gamma f_\nu
\end{equation}
is divergence free: 
\begin{equation}
\nabla^\mu K_\mu=0.   
\end{equation}
Moreover, such $\Gamma$  does not vanish anywhere and it is unique modulo multiplication by a constant. We can choose $\Gamma>0$. 
\begin{enumerate}
    \item If $\xi$ is a Killing vector and $\Lie_\xi f_\mu =0$, then $\Lie_\xi \Gamma=0.$
   \item \label{rm:trivKill} $K=0$ if and only if the one form $f$ is exact. Moreover, in such situation $f_\mu=-\nabla_\mu \ln \Gamma$
\end{enumerate}

\end{prop}

\begin{proof}
Given arbitrary function  $\Gamma$, the divergence of $K$ is given by the action on $\Gamma$ of a differential operator, namely 
   \begin{equation}\label{divK}
\Gamma\mapsto  \nabla^\mu K_\mu=\Delta \Gamma+\nabla^\mu\left(\Gamma f_\mu\right),
\end{equation}
that belongs to the class of the AMS operators considered in Theorem \ref{tm:AMS}.  The conclusion  follows from the conclusion of Theorem \ref{tm:AMS}.

If the $1$-form $f_\mu$ is exact, then Proposition \ref{prop:divfreeK} provides a trivial vector field $K=0$.    Indeed, if there is a function $B$ such that 
\begin{equation}\label{K=0}
    f_\mu = \nabla_\mu B,
\end{equation}
then  the function
\begin{equation}
 \Gamma = {\rm exp (-B)}   
\end{equation}
corresponds to the vector field $K=0$, that is in particular divergence free. Hence $\Gamma$ is the  non-trivial function provided by the theorem, however the corresponding vector field is trivial. The opposite is also true, $K=0$ implies \eqref{K=0}. Notice, that $\Gamma>0$ thus $B=-\ln \Gamma$ is well-defined.
\end{proof}

\section{Master Identity of the  EEH equation.}
The explicitly known solutions to the EEH equation   (\ref{theequation}) admit a Killing vector field $K$ that may be written in the following form  
\begin{equation}
 K_\mu = \nabla_\mu\Gamma -2\Gamma \omega_\mu. 
\end{equation}
It  justifies introducing the following notation, 
\begin{equation}
    h_\mu := -{2}\omega_\mu, \ \ \ \ \  K_\mu = \nabla_\mu \Gamma + \Gamma h_\mu
\end{equation}
In terms of $h_\mu$, the EEH  equation (\ref{theequation}) reads
\begin{equation}\label{theequation'}
\nabla_{(\mu}h_{\nu)}-\frac{1}{2}h_\mu h_\nu +R_{\mu\nu}-\lambda g_{\mu\nu}=0
\end{equation}
where 
\begin{equation}
    \lambda := \frac{n}{2}\Lambda. 
\end{equation}
The key reason for many properties of the solutions to the EEH equation is the following identity:
\begin{lm}\label{lm:distiled}{\bf Master Identity of the EEH equation.}
Suppose  $(h,g)$ satisfies the equation \eqref{theequation'} with a constant $\lambda$ on a manifold $\Sigma$. For arbitrary  function $\Gamma$ defined on $\Sigma$, define the following vector field
\begin{equation}\label{K}
K_\mu:=\nabla_\mu \Gamma+\Gamma h_\mu.
\end{equation}
The following identity holds:
\begin{equation}
\square K_\mu=(\nabla^\nu K_\nu) h_\mu+\nabla_\mu V-\Omega_{\mu\nu} K^\nu,
\end{equation}
where $\square$ is the Killing "wave" operator, $V=\Delta \Gamma+2\lambda \Gamma$ and $\Omega_{\mu\nu}:=2\nabla_{[\mu}h_{\nu]}$.
\end{lm}

For the brevity of our presentation, we postpone a proof of this identity to Appendix \ref{sec:MI-derivation}. It is based on direct computation using only \eqref{theequation'}. 

An identity obtained by contracting  Master Identity of the EEH equation  with the Dunajski-Lucietti vector field $K^\mu$ is derived in \cite{DL-rigidity} and used intensively in their work. The full  Master Identity derived above  is stronger, it contains the part orthogonal to $K$. That part will be relevant in  our new arguments and results presented in  Secs. \ref{sec:rigidity} below.     

\section{Rigidity of solutions to the EEH equation.}\label{sec:rigidity}
Combining Master  Identity satisfied by solutions of the EEH equation (Lemma \ref{lm:distiled}), with Proposition \ref{prop:divfreeK}, and with the sufficient condition for the existence of Killing vector field of Theorem \ref{tm:Killing} immediately implies the following result, which was obtained originally in \cite{DL-rigidity} and \cite{Colling2024}. Our argument based directly on Master Identity allows to simplify the proof and make it straightforward:

\begin{prop}\label{pr:rigiditya}{\bf Rigidity}
Suppose a Riemannian metric tensor $g_{\mu\nu}$ and a non-exact differential $1$-form $h_\mu$ defined on a compact manifold $\Sigma$ satisfy the  equation 
\begin{equation}
\nabla_{(\mu}h_{\nu)}-\frac{1}{2}h_\mu h_\nu +R_{\mu\nu}-\lambda g_{\mu\nu}=0
\end{equation}
with a constant $\lambda$. Then,
\begin{enumerate}
    \item The metric tensor and $1$-form admit a symmetry vector $K^\mu$,
    \begin{equation}
        \Lie_K g=0,\ \Lie_K h=0,
    \end{equation}
    where the vector field $K^\mu$, can be written in the form 
\begin{equation}
 K_\mu = \nabla_\mu\Gamma +\Gamma h_\mu, 
\end{equation} 
where the function $\Gamma$ is unique up to rescalling by constant, and it nowhere vanishes. Moreover, $\Lie_K\Gamma=0$. Vector $K$ is nonzero if and only if form $h$ is not exact.
\item Function $A$ defined by
\begin{equation}
    A:=\Delta \Gamma+2\lambda \Gamma-\frac{K^\mu K_\mu}{\Gamma},
\end{equation}
is constant.
\end{enumerate}
\end{prop}

\begin{remark}
Let us notice that $A$ is constant also in the case of exact $h_\mu$, however in this case $K^\mu=0$.
Let us notice that constancy of $A$ is exactly, what is needed in the proof of symmetry enhancement in \cite{DL-rigidity}.
\end{remark}

\begin{proof}
We consider $\Gamma$ and $K^\mu$ provided by Lucietti-Dunajski construction (it satisfies $\nabla_\mu K^\mu=0$). The master identity gives
\begin{equation}
    \square K_\mu=\nabla_\mu V+A_\mu,\quad A_\mu=-\Omega_{\mu\nu} K^\nu.
\end{equation}
Theorem \ref{tm:Killing} shows that $K^\mu$ is a Killing vector because $A_\mu K^\mu=0$. Notice that $K^\mu\not=0$ if $h_\mu$ is not exact. 

In order to prove that $\Lie_K h=0$ is actually true for arbitrary value of $\lambda$,  we make departure from the Dunajski-Lucietti proof, and present our own methods.  Taking advantage of our un-contracted identity of Lemma  \ref{lm:distiled},  knowing that $K$ is the Killing vector, hence $\square K_\nu=0$, we obtain also 
\begin{equation}
0=\nabla_\mu V-\Omega_{\mu\nu}K^\nu. 
\end{equation}
This equation can be written in the form language as
\begin{equation}
dV+K\llcorner dh=0,\quad h=h_\mu dx^\mu.
\end{equation}
For brevity let us denote in this proof $\dot{h}=\Lie_K h$.
The Cartan formula gives us
\begin{equation}\label{U}
\dot{h}:=\Lie_Kh=d\left(K\llcorner h\right)+K\llcorner dh=dU,\quad U=K_\mu h^\mu-V.
\end{equation}
Our goal is to show that the function $U$ is constant.

The Lie derivative of the EEH  equation \eqref{theequation'} gives
\begin{equation}
\nabla_{(\mu}\dot{h}_{\nu)}-h_{(\mu}\dot{h}_{\nu)}=0.
\end{equation}
The contraction after using $\dot{h}_\mu=\nabla_\mu U$ can be written in terms of an AMS-type operator (Definition \ref{AMSop})
\begin{equation}
{H'}^\dagger U=0,\quad {H'}^\dagger U:=-\Delta U+h^\mu \nabla_\mu U.
\end{equation}
Hence, we can apply now Theorem \ref{tm:AMS}  and infer that 
\begin{equation}
    U = \const .
\end{equation}
This property holds independently whether $K^\mu$ in nontrivial or not. Let us notice that Theorem \ref{tm:AMS} shows that $\Lie_K\Gamma=0$
because both metric and $1$-form is preserved by $K^\mu$.

Finally, we compute
\begin{equation}
U=K^\mu\frac{K_\mu-\nabla_\mu \Gamma}{\Gamma}-\Delta \Gamma-2\lambda \Gamma-=\frac{K^\mu K_\mu}{\Gamma}-\Delta \Gamma-2\lambda \Gamma.
\end{equation}
where we used $K^\mu\nabla_\mu \Gamma=0$. This shows that $A=-U$ is constant.
\end{proof}

\section{Static solutions of the EEH equation.}
A special class of solutions  $(\Sigma, g_{\mu\nu},h_\mu)$ to the EEH equation (\ref{theequation'}) for a compact manifold $\Sigma$ and  Riemannian metric tensor $g_{\mu\nu}$,  are non-rotating solutions (also called "static"), namely such that
\begin{equation}
    \Omega_{\mu\nu} := 2\nabla_{[\mu} h_{\nu]} = 0. 
\end{equation}
Our Master Identity of Lemma \ref{lm:distiled} is useful also in this  case. We notice that if $\Omega_{\mu\nu}=0$, then Master Identity implies
\begin{equation}\label{DeltaGamma}
0=\square K_\nu=\nabla_\nu V\Longrightarrow \Delta \Gamma+2\lambda \Gamma=\const.
\end{equation}
Let us focus now on the case
\begin{equation}\label{lambdale0}
\lambda\leq 0.     
\end{equation}
The equality (\ref{DeltaGamma}) implies for the positive defined function $\Gamma$ given by Proposition \ref{prop:divfreeK}, that 
\begin{equation}\label{hck}
    \Gamma=c=\const>0 \ \ \ \ {\rm and}\ \ \ \ h_\mu=c^{-1}K_\nu.
\end{equation}
Since $K^\mu=0$ when $h_\mu$ is exact, it follows that 
\begin{equation}
h_\mu \ {\rm exact}\ \ \Rightarrow \ \ h_\mu=0.     
\end{equation}
\begin{prop}
The general solution $(g_{\mu\nu},h_\mu)$ to the EEH equation (\ref{theequation'}) on a compact manifold with $\lambda\le 0$ and  $h_\mu$ exact, is such that
\begin{equation}
    h_\mu = 0, \ \ \ \ {\rm and} \ \ \ \ R_{\mu\nu}=\lambda g_{\mu\nu}. 
\end{equation}
\end{prop}

Now, let us assume, that $h_\mu$ is closed however not exact. Then $K_\mu$ in (\ref{hck}) is not identically zero, hence $h^\mu$ becomes a Killing vector field itself making it covariantly constant,  
\begin{equation}
\nabla_\mu h_\nu=0.
\end{equation}
A Riemannian space that admits a covariantly constant vector field  $h^\mu$ is locally the Cartesian product of (orthogonal to each other) another Riemannian space $(\Sigma', g'_{\mu\nu})$ and a line $(\R, g_{vv}dv^2)$,
\begin{equation}\label{eq:splitting}
\Sigma=\Sigma'\times \R,\quad g=g'\oplus g_{vv}dv^2,\ h=0\oplus dv.
\end{equation}
To learn more about $g'_{\mu'\nu'}$ and $g_{vv}$, let us go back to the EEH equation, that reads now
\begin{equation}
R_{\mu\nu}=\frac{1}{2}h_\mu h_\nu +\lambda g_{\mu\nu},
\end{equation}
and splits into the $\Sigma'$ and $\R$ parts, namely
\begin{equation}\label{eq:splitting'}
R'_{\mu'\nu'}=\lambda g'_{\mu'\nu'}, \ \ \ \ {\rm and}\ \ \ \ \ 0=\frac{1}{2} +\lambda g_{vv}.
\end{equation}
Notice, that in particular it follows, that for $h_\mu$ closed however not exact,  the inequality  (\ref{lambdale0}) should be replaced by
\begin{equation}\label{lambda<0}
\lambda< 0,     
\end{equation}
and $g_{vv}=-\frac{1}{2\lambda}$.

\subsection{General static solution}

In general, however, the  splitting  (\ref{eq:splitting},\ref{eq:splitting'}) is not global. Let us therefore turn, to the global analysis.  
Let $\phi_t:\Sigma\rightarrow\Sigma$ denotes the flow generated by $h^\mu$.

\begin{lm}
If $h_\mu$ is non-trivial then there exists $T\not=0$ such that $\phi_T$ is identical. In particular, every orbit is a closed circle.
\end{lm}

\begin{proof}
The group of isometries $G$ of a compact manifold is a compact Lie group  (see \cite{Kobayashi1995}). In particular the flow $\phi_t$ is either periodic or the maximal torus in $G$ is at least two dimensional (and $\phi_t$ is winding in the torus incommensurable). In the second case there exists another Killing vector $F$ such that $[F,h]=0$. We would like to exclude this case.  We will do it by application of the Bochner technique \cite{Bochner46}.

Let us notice that
\begin{equation}
\nabla_\nu(h^\mu F_\mu)=h^\mu \nabla_\nu F_\mu=-h^\mu \nabla_\mu F_\nu+F^\mu \nabla_\mu h_\nu=0,\Longrightarrow c:=h^\mu F_\mu=\const
\end{equation}
We can replace $F$ by another Killing vector ${F'}^\mu=F^\mu-\frac{c}{\sqrt{-2\lambda}}h^\mu$, that is also commuting with $h^\mu$. We can thus assume $c=F^\mu h_\mu=0$.

We contract the EEH equation with two $F^\mu$ vectors obtaining
\begin{equation}
R_{\mu\nu} F^\mu F^\nu=\lambda F_\mu F^\mu.
\end{equation}
Killing equation provides
\begin{equation}
0=F^\mu \square F_\mu=F^\mu \Delta F_\mu+\lambda F_\mu F^\mu.
\end{equation}
Integrating this formula over $\Sigma$ (and integrating by parts first term) we obtain
\begin{equation}
0=-\int_\Sigma \nabla_\mu F_\nu \nabla^\mu F^\nu-\lambda F_\mu F^\mu.
\end{equation}
As $\lambda<0$ for $h_\mu\not\equiv0$, the integrand is nonnegative thus identically zero. This means that $F^\mu=0$.

The flow of $\phi_t$ is thus periodic and the orbits are closed circles.
\end{proof}

Let us consider a static solution $(\Sigma,g,h)$ with $h\not=0$. For a closed form we can define a homomorphism from the first homotopy group $\pi_1(\Sigma)$ into group $\R$
\begin{equation}
    \Psi_\R\colon \pi_1(\Sigma)\rightarrow \R,
\end{equation}
by the following formula
\begin{equation}
    \Psi(m)=\int_{\gamma_m}h,
\end{equation}
where $\gamma_m$ is any close loop representant of $m\in \pi_1(\Sigma)$.

We define a particular element in $\pi_1(\Sigma)$. Let $T_h>0$ be a minimal $t\in \R_+$ such that $\phi_t={\mathbb I}$. For every $x\in \Sigma$,
\begin{equation}
    \gamma_h\colon [0,T_h]\ni t\rightarrow \phi_t(x)
\end{equation}
is a closed loop. Its class $[\gamma_h]$ in $\pi_1(\Sigma)$ does not depend on a choice of $x$.

\begin{lm}
The subgroup of $\pi_1(\Sigma)$ generated by $[\gamma_h]$ is central and isomorphic to ${\mathbb Z}$.
\end{lm}

\begin{proof}
Let $\gamma\colon [0,1]\rightarrow \Sigma$ be a closed loop $\gamma(0)=\gamma(1)=x$. Let us consider a map
\begin{equation}
    [0,1]\times [0,T_h]\ni (s,t)\rightarrow \phi_t(\gamma(s))\in \Sigma
\end{equation}
The boundary curve $\partial ([0,1]\times [0,T_h])\rightarrow \Sigma$ (starting at $(0,0)$) belongs to the homotopy class $[\gamma_h][\gamma][\gamma_h]^{-1}[\gamma]^{-1}$. Thus the map provides a contraction and the class is trivial,  namely $[\gamma_h][\gamma]=[\gamma][\gamma_h]$.

In order to show that the group generated by $[\gamma_h]$ is isomorphic to ${\mathbb Z}$ it is enough to show that for every $n\not=0$, $[\gamma_h]^n\not=1$. The closed curve
\begin{equation}
    \gamma'\colon [0,nT_h]\ni t\rightarrow \phi_t(x)\in \Sigma
\end{equation}
is a representant of the class $[\gamma_h]^n$. We compute
\begin{equation}
    \Psi_\R([\gamma_h]^n)=\int_{\gamma'}h=\int_0^{nT_h}h_\mu h^\mu dt=-2\lambda nT_h
\end{equation}
It shows that $[\gamma_h]^n\not=1$.
\end{proof}

We denote by ${\mathbb Z}_h$ a subgroup generated by $[\gamma_h]$. We define
\begin{equation}
    \Psi_h\colon \pi_1(\Sigma)/{\mathbb Z}_h\rightarrow U(1),\quad \Psi_h([\gamma])=-\frac{\pi}{\lambda T_h}\Psi_\R([\gamma])\text{ mod } 2\pi
\end{equation}
where we used fact that $\Psi_\R([\gamma_h])=-2\lambda T_h$. We use a convention identifying $U(1)$ with the interval $[0,2\pi)$. This map measures departure from the global splitting. If $\Sigma=\Sigma'\times S^1$ (with the product solution) then $\Psi_h$ is trivial, namely it maps every elements to $0\in U(1)$.

We will now describe the most general static solution.
Let $(\tilde{\Sigma}',g')$ be a simply connected Einstein solution with $\lambda<0$. We define a metric $\tilde{g}$ and one form $\tilde{h}$ on $\tilde{\Sigma}=\tilde{\Sigma}'\times \R$ (see \eqref{eq:splitting})
\begin{equation}\label{tilde-sigma}
\tilde{\Sigma}=\tilde{\Sigma}'\times \R,\quad \tilde{g}=g'\oplus -\frac{1}{2\lambda}dv^2,\ \tilde{h}=0\oplus dv.
\end{equation}
We denote by $\tilde{G}'$ the group of isometries of $\tilde{\Sigma}'$. It is discrete group (see \cite{Bochner46, Kobayashi1995}).

We choose a subgroup $G'\subset \tilde{G}'$ and a homomorphism
\begin{equation}
    \psi\colon G'\rightarrow U(1)
\end{equation}
such that
\begin{enumerate}
    \item $G'$ acts cocompactly on $\tilde{\Sigma}'$ (that is $\tilde{\Sigma}'/G'$ is compact)
    \item If for some $m\in G'$  there exists a fixed point $x\in \tilde{\Sigma}'$ then $\psi(m)\not=0$,
\end{enumerate}
We choose also $T>0$. We will call data $(\tilde{\Sigma}',G',\psi,T)$ admissible if the above conditions are satisfied.

Having admissible data $(\tilde{\Sigma}',G',\psi,T)$ we construct $\Sigma$ as follows. We divide $\tilde{\Sigma}$ by the group
\begin{equation}
    \tilde{G}=\left\{(m,t)\colon m\in G',\ t=-\frac{\lambda T}{\pi}(\psi(m)+2\pi n),\ n\in {\mathbb Z}\right\}\subset \tilde{G}'\times \R,
\end{equation}
where the action of $(m,t)$ on $\tilde{\Sigma}$ is given by
\begin{equation}
    (m,t)(x,r)=(m(x),r+t).
\end{equation}
We denote by $\pi_{\tilde{G}}\colon \tilde{\Sigma}\rightarrow \Sigma$  the quotient projection.

\begin{lm}\label{quotient}
For admissible data $(\tilde{\Sigma}',G',\psi,T)$, the resulting quotient space $\Sigma$ is equipped with metric and closed one form satisfying EEH equation and such that
\begin{equation}
    \pi_{\tilde{G}}^*g=\tilde{g},\quad \pi_{\tilde{G}}^*h=\tilde{h}.
\end{equation}
Moreover, $T_h=T$, the group $\pi_1(\Sigma)/{\mathbb Z}_h$ is isomorphic to $G'$ and under this isomorphism $\Psi_h=\psi$.
\end{lm}

\begin{proof}
The action of $\tilde{G}$ preserves both the metric $\tilde{g}$ and one form $\tilde{h}$. Moreover, the action is free, thus the resulting quotient space $\Sigma$ is a manifold equipped with quotient metric and quotient one form.

We would like to show that $\Sigma$ is compact. Let $y_n$ be a sequence of points in $\Sigma$. There exist $(x_n,t_n)\in \tilde{\Sigma}$ such that $\pi_{\tilde{G}}(x_n,t_n)=y_n$. The group $G'$ acts cocompactly on $\tilde{\Sigma}'$ thus we can choose $x_n$ is such way that there exists a subsequence $x_{n_i}$ convergent to a point in $\tilde{\Sigma}'$. Let us notice that $t_n$ are not unique but can be changed by $-2n\lambda T$. In particular we can assume that $t_n\in [0,-2\lambda T)$. There exists a converging subsequence of $t_{n_i'}$ of $t_{n_i}$. A sequence $(x_{n_i'},t_{n_i'})$ is convergent, thus also $y_{n_i'}=\pi_{\tilde{G}}(x_{n_i'},t_{n_i'})$ is convergent. This shows compactness.

Let us notice that on $\tilde{\Sigma}$ the flow $\tilde{\phi}_t$ generated by the vector $\tilde{h}=-2\lambda \partial_v$ is equal
\begin{equation}
    \tilde{\phi}_t(x,r)=(x,r-2\lambda t).
\end{equation}
Moreover, $\phi_t={\mathbb I}$ if only if $\tilde{\phi}_t\in {\tilde{G}}$. The conditions for admissible data shows that this is equivalent to $t=nT$ and thus $T_h=T$.

Finally, the first homotopy group is isomorphic to $G$ (by construction) and
\begin{equation}
    \Psi_R(m,t)=t.
\end{equation}
Moreover ${\mathbb Z}_h=\{(0,nT)\colon n\in {\mathbb Z}\}$. This shows that $\Psi_h=\psi$.
\end{proof}

We will now show that this is the most general solution.

\begin{thm}\label{thm:generalstatic}
Every static solution of EEH equation on a compact manifold with $\lambda<0$ and such that $h\not=0$ is diffeomorphic to the quotient space of Lemma \ref{quotient} for admissible data. 
\end{thm}

\begin{proof}
Consider the universal cover $\tilde{\Sigma}$ of $\Sigma$ with the pull-backed metric $\tilde{g}$ and $1$-form $\tilde{h}$. The  fundamental group $G=\pi_1(\Sigma)$ can be identified with subgroup of diffeomorphisms of $\tilde{\Sigma}$. This subgroup preserves $\tilde{g}$ and $\tilde{h}$ and its action is free and co-compact.

Global splitting shows that $(\tilde{\Sigma},\tilde{g},\tilde{h})$ is of the form \eqref{tilde-sigma}. 
The action of the group $G$ preserves the splitting (it preserves $dv$) thus for $p\in G$
\begin{equation}
    p(x,r)=(m_px,r+t_p)
\end{equation}
where $t_p\in \R$ and $m_p$ is a isometry of $\tilde{\Sigma}'$. We can identify $G$ with subgroup of the product group of isometries of $\tilde{\Sigma}'$ ($\operatorname{ISO}(\tilde{\Sigma}')$) and $\R$. We remark that the group $\operatorname{ISO}(\tilde{\Sigma}')$ is discrete \cite{Bochner46}.

Let us notice that as $\phi_{T_h}={\mathbb I}$ we have
\begin{equation}
    \tilde{\phi}_{T_h}=(0,-2\lambda T_h)\in G,
\end{equation}
and ${\mathbb Z}_h=\{({\mathbb I},-2\lambda nT_h)\colon n\in {\mathbb Z}\}$.
Moreover, if $(0,t)\in G$ then $t$ is the multiplicity of $-2\lambda T_h$ (from minimality of $T_h$). Consequently, the map
\begin{equation}
    G':=G/{\mathbb Z}_h\rightarrow \operatorname{ISO}(\tilde{\Sigma}')
\end{equation}
is injective and
we can identify $G'$ with a subgroup of isometries of $\tilde{\Sigma}'$. We have a well-defined map
\begin{equation}
    \psi\colon G'\rightarrow U(1)
\end{equation}
defined by $\psi(m,t)=-\frac{\pi}{\lambda T_h}t$ modulo $2\pi$. We can write
\begin{equation}
    G=\left\{(m,t)\colon m\in G',\ t=-\frac{\lambda T}{\pi}(\psi(m)+2\pi n),\ n\in {\mathbb Z}\right\}.
\end{equation}
Moreover, $\Psi_h=\psi$. We notice that condition for free action is exactly
\begin{equation}
    p(x)=x\Longrightarrow m_p={\mathbb I} \text{ and } \psi(m_p)=0
\end{equation}
Additionally, as $G$ acts co-compactly on $\tilde{\Sigma}$ the same is true for the action of $G'$ on $\tilde{\Sigma}'$.

We see that $(\tilde{\Sigma}',G',\psi, T_h)$ is admissible and $\Sigma$ is the quotient space for this data.
\end{proof}

Thus far in this section, we have considered non-positive cosmological constant. For 
\begin{equation}
    \Lambda >0,
\end{equation}
there exist only partial results \cite{NHG} (see also \cite{Wylie2023} for the discussion).

\section{ Generalized  extremal horizon equations  on surfaces (case  \texorpdfstring{$n=2$}{n=2})}\label{sec:case2}
In this section we consider a $2$ dimensional Riemann space $(\Sigma, g_{\mu\nu})$.  We outline first, the results on the EEH equation on a compact $\Sigma$.  In this case, all the solutions to the  EEH equation are known. To start with, for arbitrary value of the cosmological constant $\lambda$,  the only static solutions  on $S^2$ are those of 
\begin{equation}\label{omega0}
  \omega_\mu=0, \ \ \ \ \text{and}\ \ \ \  \frac{1}{2}R = \lambda,   
\end{equation}
 where $R=g^{\mu\nu}R_{\mu\nu}$ is the Ricci scalar \cite{CRT}, hence of a constant  curvature. Next, if the Euler invariant $\chi_E(\Sigma)\le 0$, then again, the only solution  is (\ref{omega0}). 
On the other hand,   if $\Sigma$ is diffeomorphic to $S^2$ and $\omega_\mu$ is not exact, then Rigidity Theorem implies axial symmetry of the metric tensor $g_{\mu\nu}$ and the rotation $1$-form $\omega_\mu$. However, all the axially symmetric solutions are explicitly know, they correspond to extreme horizons in the Kerr-(A)dS spacetimes  \cite{LPextremal,Kunduri:2008tk,KunduriVac,Buk:2020ttx}. That covers all the solutions.  

In this section, we consider generalized equations of interest: 
the Einstein-Maxwell extreme horizon (EMEH) equation  Definition \ref{etremmaxeq} and quasi-Einstein spaces. In fact, our results will be applicable to a yet larger family of equations  (generalized extremal horizon equations GEH, Definition \ref{etremgeq}), namely  $g_{\mu\nu}$ and a $1$-form $\omega_\mu$, that fulfill the following:
\begin{equation}\label{generalizedequation}
    \nabla_{(\mu} \omega_{\nu)} + a \omega_\mu \omega_\nu + f g_{\mu\nu}=0, \ \ \ \ \ a = \const\not=0, \ \ \ f\in C(\Sigma),
\end{equation}
where $a$ and $f$ are arbitrary. Remember that in $2$ dimensions the Ricci tensor 
\begin{equation}
    R_{\mu\nu} = \frac{1}{2}R g_{\mu\nu},
\end{equation}
for this, both, the EEH equation Definition  \ref{extreq} and its EMEH generalization Definition \ref{etremmaxeq}  fall into this category, as does, finally, the quasi Einstein metric equation for  $a=\frac{2}{m}$ and $f=-\frac{1}{4}R+\frac{1}{2}\Lambda$. Another generalization we admit,  is relaxing the compactness assumption in the second part of this section.

Before starting our investigation we need to invoke  structures that are available on $2$ dimensional Riemannian spaces.  Non orientable $2$ dimensional Riemannian spaces admit (at most two fold) orientable coverings that is also compact.  The  equations we consider can be pulled back and solved on the covering space. Therefore, consider now oriented Riemann space $(\Sigma,g_{\mu\nu})$. 
 For every point $x\in\Sigma$, let us split at each $x\in \Sigma$ the complexified tangent space
\begin{equation}
 \mathbb{C}\otimes T_x\Sigma =  T^{(1,0)}_x\Sigma \oplus  T^{(0,1)}_x\Sigma, 
\end{equation}
using the two form
\begin{equation}
P_{\mu\nu}=\frac{1}{2}(g_{\mu\nu}+i\epsilon_{\mu\nu}),
\end{equation}
where $\epsilon_{\mu\nu}$ is the $2$-volume tensor,  defining the decomposition for arbitrary vector $X\in T_x\Sigma$,
\begin{equation}
X^{(1,0)\mu} := P^\mu{}_\nu X^\nu, \ \ \ \ X^{(0,1)\mu} :=X^\nu {P}_\nu{}^\mu .  
\end{equation}
The corresponding decomposition for the dual space is 
\begin{equation}
   \mathbb{C}\otimes T_x^*\Sigma =  T^{*(1,0)}_x\Sigma \oplus  T^{*(0,1)}_x\Sigma, 
\end{equation}
\begin{equation}
Y^{(1,0)}_\mu :=Y_\nu  P^\nu{}_\mu , \ \ \ \ Y^{(0,1)}_\mu :=  {P}_\mu{}^\nu  Y_\nu 
\end{equation}
 An observation to remember is that 
\begin{equation}
    Y_\mu \in  T^{*(0,1)}_x\Sigma \Rightarrow Y^\mu \in  T^{(1,0)}_x\Sigma. 
\end{equation}
We also use the decomposition to introduce the anti-holomorphic covariant derivative, 
\begin{equation}
D_\mu:=P_\mu^{\phantom{\mu}\nu}\nabla_\nu.   
\end{equation}
Note,  that $P$ commutes with the covariant derivative. 

A function $\phi\in C(\Sigma,\mathbb{C})$ that satisfies
\begin{equation}
D_\mu \phi = 0,    
\end{equation}
is called holomorphic. 

A $1$-form $\psi_\mu$ of the type $(1,0)$ is called holomorphic, if 
\begin{equation}
    D_\mu \psi_\nu = 0.
\end{equation}
If a $1$-form $Y_\mu$ of the $(0,1)$ type satisfies
\begin{equation}
    D_\mu Y_\nu = 0,
\end{equation}
then the vector field $Y^\mu$ is called holomorphic.

The double projection of the equation (\ref{generalizedequation}) onto the $T^{*(0,1)}_x\Sigma \otimes T^{*(0,1)}_x\Sigma$ component (in other words: the trace free part) at every $x\in \Sigma$, results in the following:
\begin{equation}\label{dpi}
 D_{\mu} \pi_{\nu} + \pi_\mu\pi_\nu =0, \ \ \ \ \ \pi_\mu := a P_\mu{}^\nu \omega_\nu.
\end{equation}

\subsection{Topological censorship}
This chapter is devoted to the following result and its applications: 
\begin{prop}\label{genus}{\bf Topological censorship}
Suppose $(g_{\mu\nu}, \omega_\mu)$ satisfy the   generalized extremal horizon equation (GEH) with arbitrary $a$ and $f$. If $\Sigma$ is compact and 
its Euler invariant is non positive
\begin{equation}
    \chi_E(\Sigma)\le 0,
\end{equation}
then 
\begin{equation}\label{eq:topcenz}
\omega_\mu = 0 = f.    
\end{equation}
\end{prop}
It is sufficient to prove (\ref{eq:topcenz}) for an orientable $\Sigma$ of non-zero genus. Deriving a holomorphic vector field from $\pi_\mu$ and $g_{\mu\nu}$ that satisfy (\ref{dpi}) is our strategy.

\begin{proof}
The proof will be based on the following result:

\begin{lm}\label{lm:df}
Let $\Sigma$ be Riemannian surface of genus $g>0$. Suppose that $f_\mu$ satisfies
\begin{equation}
 D_\mu f_\nu+f_\mu f_\nu=0,\quad f_\mu=P_{\mu\nu}f^\nu
\end{equation}
then $f_\mu\equiv 0$.
\end{lm}

\begin{proof}
Suppose that $f_\nu$ is not identically zero. As surface has genus greater than zero there exist a nontrivial  holomorphic $1$-form $\psi_\mu$. We remark that $\psi_\mu$ has only isolated zeros.

Let us consider $\phi=f_\mu \psi^\mu$. Locally
\begin{equation}
f_\mu dx^\mu= F \rd\overline{z},\quad \psi_\mu dx^\mu=P\rd z,
\end{equation}
and $P$ has only isolated zeros. Thus if $f_\mu$ is nontrivial then $\phi$ is also nontrivial. Let us now notice that $D_\mu \psi_\nu=0$ so
\begin{equation}
D_\mu \phi=D_\mu(f_\nu \psi^\nu)=(D_\mu f_\nu) \psi^\nu=-f_\mu f_\nu \psi^\nu=-f_\mu \phi.
\end{equation}
As proved in \cite{DKLS2} Lemma 1, $\phi$ is everywehere nonvanishing. Now,
\begin{equation}
D_\mu (\phi^{-1} f_\nu)=-\phi^{-2}(D_\mu \phi)f_\nu+\phi^{-1}D_\mu f_\nu=\phi^{-2} f_\mu \phi f_\nu-\phi^{-1} f_\mu f_\nu=0
\end{equation}
This means (see \cite{DKLS1, Griffiths}) that $\phi^{-1} f^\nu$ is a homolophic vector field.
Repeating the argument from \cite{DKLS1, DKLS2}: On surfaces of genus bigger than one there is no nontrivial holomorphic vector fields thus $f_\nu\equiv 0$ (contradiction).

The case of torus (genus one) require a separate treatment, because holomorphic vector fields exists. Let us notice that
\begin{equation}
    \phi^{-1} f_\mu=-\phi^{-2}D_\mu \phi=D_\mu \phi^{-1}.
\end{equation}
However, the holomorphic vector fields on the torus are constant and they cannot be written as $D^\mu$ of a function (see \cite{DKLS1, DKLS2} for a similar argument).
\end{proof}

In view of (\ref{dpi}), Lemma \ref{lm:df} completes the proof of Proposition \ref{genus} for an orientable $\Sigma$. If $\Sigma $ is not orientable, we can pull back the equation (\ref{generalizedequation}) to the covering orientable space, and conclude that the pullback of $\pi_\mu$ is trivial, then so is $\pi_\mu$ itself.     
\end{proof}

Let us now apply our very findings to specific examples.

\begin{cor}
On compact, $2$ dimensional Riemannian space $(\Sigma,g_{\mu\nu})$ of a non-positive Euler invariant $\chi_E(\Sigma)\le 0$,  for arbitrary  values of the parameters $a\not=0$ and $\Lambda$,  the general solution to  an equation
\begin{equation}
\nabla_{(\mu}\omega_{\nu)} + a\omega_\mu\omega_\nu -\frac{1}{2}R_{\mu\nu}+ \frac{1}{2}\Lambda g_{\mu\nu} =0 
\end{equation}
is 
\begin{equation}
    \omega_\mu = 0, \ \ \ \ {\rm and}\ g_{\mu\nu}  \ {\rm such\ that} \ R_{\mu\nu}=\Lambda g_{\mu\nu}. 
\end{equation}
\end{cor}

\begin{cor}\label{cor:EMEHgenus}
On compact, $2$ dimensional Riemannian space $(\Sigma,g_{\mu\nu})$ of a non-positive Euler invariant $\chi_E(\Sigma)\le 0$,  for arbitrary  values of the cosmological constant $\Lambda$,  the general solution to the EMEH equations (Definition \ref{etremmaxeq})   is
\begin{equation}
    \omega_\mu = 0, \ \ \ \ \Phi=\const, \ \ \ \  {\rm and}\ g_{\mu\nu}  \ {\rm such\ that} \ R_{\mu\nu}=(\Lambda +  4|\Phi|^2)\, g_{\mu\nu}. 
\end{equation}
\end{cor}
Indeed, notice, that as $\omega_\mu=0$, the second equation in Definition \ref{etremmaxeq} boils down to
\begin{equation}
    D_\mu \Phi = 0,
\end{equation}
hence $\Phi$ is holomorphic, and in the consequence constant. 

\subsection{Static rigidity.}

Given the results of the previous subsection on $2$ dimensional $\Sigma$ of non-positive Euler invariant, the problem  we are left with is $\Sigma$ (or its double cover) diffeomorphic to $S^2$. We will address in this section a case of exact $\omega_\mu$ and show that in this situation GEH equation posses a symmetry. Remarkably,  the result that  will be presented does not even require the compactness of $\Sigma$. For every simple connected surface (like $S^2$), every closed form is exact.

\begin{prop}\label{staticrigidity}{\bf Static rigidity}
Suppose $(g_{\mu\nu}, \omega_\mu)$ satisfy the GEH equation (\eqref{generalizedequation} with arbitrary $a$ and $f$). If  $\omega_\mu$ is exact however is not identically $0$, then there exists a non-zero vector field $K$, such that
\begin{equation}
    \Lie_K g_{\mu\nu} = 0 = \Lie_K \omega_\nu .
\end{equation}
\end{prop}

\begin{remark}
This proposition systematizes the standard approach to staticity problem \cite{Chrusciel:2005pa, Kunduri:2008tk}, where the symmetry appears a posteriori through judicious choice of the coordinate system. A recent important work \cite{CKL-new} shows that also every non-static solution possessed a symmetry.
\end{remark}

The axisymmetric solutions on the sphere were classified for EEH and EMEH equations in \cite{Hajicek, LPextremal, Chrusciel:2005pa, Kunduri:2008tk}.  Apparently the case of electromagnetic field with positive cosmological constant was not considered in these works. We will show uniqueness of Reissner-Nordstrom-de Sitter in the next section.

We start with the observation due to \cite{LS}:

 \begin{lm}\label{lm:static-DD}
 Let real form $\omega_\mu$ be exact but not identically zero. Let $\phi=e^{\psi}>0$ where $\omega=d\psi$ ($\psi$ is real). We introduce vector field  $K^\mu=-\epsilon^{\mu\nu}\partial_\nu\phi$. The following two conditions are equivalent
 \begin{enumerate}
     \item $K^\mu$ is a Killing vector
     \item $D_\mu f_\nu+f_\nu f_\mu=0$ where $f_\mu=P_{\mu\nu}\omega^\nu$.
 \end{enumerate}
 Moreover, $\Lie_K\omega=0=\Lie_K\phi$.
 \end{lm}

 \begin{remark}
 This result holds on arbitrary (even non-compact) surface.
 \end{remark}

 \begin{proof}
A simple calculation shows that
\begin{equation}
D_\mu D_\nu \phi=0\Longleftrightarrow D_\mu D_\nu \psi + D_\mu \psi D_\nu \psi = 0. 
\end{equation}
Remind that $f_\mu=D_\mu \psi$.
Moreover,  the functions $\psi$ and $\phi$ are real. Therefore, 
\begin{equation}
D_\mu D_\nu \phi=0 \Longleftrightarrow \overline{D}_\mu \overline{D}_\nu \phi=0.
\end{equation}
Our definition of $K^\mu$ is equivalent to
\begin{equation}
K^\mu = i\left(D^\mu \phi-\overline{D^\mu} \phi\right),
\end{equation}
then due to $D_\mu\overline{D}_\nu \phi=\overline{D}_\nu D_\mu \phi$,
\begin{align}
2\nabla_{(\mu} K_{\nu)}  &= i\left((D_\mu + \bar{D}_\mu)(D_\nu - \bar{D}_\nu) + (D_\nu + \bar{D}_\nu)(D_\mu - \bar{D}_\mu)\right)\phi = \nonumber\\
&=2i\left(D_\mu D_\nu \phi-\bar{D}_\mu\bar{D}_\nu\phi\right)
\end{align}
This means that $\nabla_{(\mu} K_{\nu)}=0$ if and only if $D_\mu f_\nu+f_\mu f_\nu=0$.

Let us also compute
\begin{align}
\Lie_K \phi&=i\left(D^\mu \phi-\overline{D^\mu} \phi\right) \nabla_\mu \phi\\
&=i\left(D^\mu \phi-\overline{D^\mu} \phi\right) \left(D_\mu \phi+\overline{D_\mu} \phi\right) =i D^\mu \phi \overline{D_\mu} \phi-i \overline{D^\mu} \phi D_\mu \phi=0.
\end{align}
In the consequence,  we also have 
$\Lie_K \omega=0$.
\end{proof}

Proposition \ref{staticrigidity} is a simple corollary from this lemma.

\begin{proof}[Proof of Proposition \ref{staticrigidity}]
Since $\omega_\mu$ is closed,  we can define a real valued  function $\psi$  such that $\nabla_\mu \psi= a \omega_\mu$ and consequently
\begin{equation}
D_\mu \psi= \pi_\mu,
\end{equation}
where $\pi_\mu$ is defined in the second part of (\ref{dpi}). Then, Lemma \ref{lm:static-DD} shows that $K^\mu=-\epsilon^{\mu\nu}\partial_\nu \phi$, $\phi=e^\psi$ is a Killing vector preserving also $\omega_\mu$.
\end{proof}

We are in the position now, to apply our very result. 

Let us start with the EEH equation Definition \ref{extreq}. It follows, that every static solution on $\Sigma$ diffeomorphic to $S^2$  is axially symmetric. However, the axisymmetric solutions are all explicitely known and the static one is $\omega_\mu=0$ for  positive value   of $\Lambda$.  This result is well-known \cite{Chrusciel:2005pa}.

Secondly, we can apply Proposition \ref{staticrigidity}  to a $m$-quasi Einstein equations (one of GEH).

\begin{cor}
Every static solution $(g_{\mu\nu},\omega_\mu )$ on $\Sigma$ diffeomorphic to $S^2$ to  the equations:
\begin{equation}
    \nabla_{(\mu} \omega_{\nu)}+ \frac{2}{m}\omega_\mu\omega_\nu +  \frac{1}{2}R_{\mu\nu} - \frac{1}{2}\lambda g_{\mu\nu} = 0 
\end{equation}
is axially symmetric.
\end{cor}
The advantage is, that in the axially symmetric case the equation amounts to a soluble system of ordinary differential equations with suitable conditions on the poles.

One more example is the system of EMEH equations Definition \ref{etremmaxeq}.

\begin{cor}\label{cor:staticrigidityofEM} Every  solution $(g_{\mu\nu}, \omega_\mu, \Phi)$ to the EMEH equations defined  on a $2$ dimensional manifold, namely
\begin{align}
       \nabla_{(\nu}\omega_{\mu)} + \omega_\mu\omega_\nu - \frac{1}{2}R_{\mu\nu} + 
        \frac{1}{2}\Lambda g_{\mu\nu} + |\Phi|^2 g_{\mu\nu} &=0,\\
        \left(D_\alpha + 2{ P_\alpha^{\phantom{\alpha}\beta}\omega_\beta} \right)\Phi &= 0,
   \end{align}
such that $\omega$ is exact,  but non-zero, admits a Killing vector $K$ such that 
\begin{equation}
    \Lie_K\omega = 0 = \Lie_K\Phi.
\end{equation}
Every static solution $(g_{\mu\nu}, \omega_\mu, \Phi)$  with $\omega\not=0$ to the EMEH equations defined on (topological) $2$-sphere is axially symmetric.
\end{cor}   

Notice, that we rewrote the Maxwell equation in terms of the anti-holomorphic derivative. Indeed, if $\omega_\mu$ is stationary, then there is an axially symmetric function $\psi'$, such that
\begin{equation}
    \omega_\mu = \nabla_\mu \psi'.
\end{equation}
The general solution to the Maxwell equation (the second one above), is
\begin{equation}
\Phi = A e^{-2\psi'}, \ \ \ \ \ \ \ D_\mu A = 0. 
\end{equation}
The second equation implies 
\begin{equation}
 A = \const. 
\end{equation}
But it also follows from the EMEH equations, that  $\Lie_K|\Phi^2|=0$, hence
\begin{equation}
    \Lie_K\psi'=0.
\end{equation}

\subsection{Static axisymmetric EMEH equations on the two dimensional sphere}

\begin{prop}\label{prop:staticEMEH}
The only static solutions $(g_{\mu\nu}, \omega_\mu, \Phi)$ to the EMEH equations defined on (topological) $S^2$, are
\begin{equation}
\omega_\mu = 0, \ \ R = \Lambda + 2|\Phi|^2, \ \ \ \Phi = \const. 
\end{equation}
\end{prop}

\begin{proof}
We will use coordinate system $(x,\psi)$
\begin{equation}
g=\frac{1}{H(x)}dx^2+H(x)d\psi^2,
\end{equation}
where $x\in [x_-,x_+]$ and $\psi\in [0,L]$ where $L>0$ (cyclic). The condition for the metric to be smooth in points $x_\pm$:
\begin{enumerate}
    \item $H(x_\pm)=0$ and $H(x)>0$ otherwise,
    \item $H'(x_-)=-H'(x_+)=\frac{4\pi}{L}$
    \item Function $H$ extends smoothly from the interval $[x_-,x_+]$
\end{enumerate}
Moreover, a function $f(x)$ is smooth on the sphere if and only if $f(x)$ extends smoothly from the interval $[x_-,x_+]$.

We choose normalization of coordinate system where
\begin{equation}
    x_\pm=x_0\pm 1.
\end{equation}
With this choice the area $A$ of the horizon is equal to $2L$.

We assume that $\omega=d\phi$ and we define $B=e^{\phi}>0$. The Killing vector field $\epsilon^{\mu\nu}\partial_\nu B$ is proportional to $\partial_\psi$, thus
\begin{equation}
    \dot{B}=c_0=\operatorname{const}.
\end{equation}
Let us remind that $B$ is a function of $x$ only and $\dot{B}$ denotes derivative in $x$.

We consider a case when $\omega\not=0$ thus we assume that $c_0\not=0$,
\begin{equation}
    B=c_0(x-x_p)
\end{equation}
By shifting variable we can choose $x_p=0$. By applying if necessary a reflection $x\rightarrow -x$ we can assume that $c_0>0$. Let us notice that $B$ is defined up to multiplication by positive constant, thus we also choose $c_0=1$ and $B=x$.  Let us remark that in this case $x_+>x_->0$ ($x_0>1$).

The conditions of EMEH equations are equivalent to
\begin{equation}\label{eq:trace}
    \nabla^\mu \omega_\mu+\omega^\mu \omega_\mu -\frac{1}{2}R+\Lambda+2|\Phi|^2=0
\end{equation}
and two conditions
\begin{align}
& D_\mu \pi_\mu+\pi_\mu \pi_\nu=0\\   
&    D_\mu \Phi+2\pi_\mu \Phi=0
\end{align}
The last two conditions are equivalent to $K^\mu$ being a Killing vector and condition
\begin{equation}
    D_\mu (\Phi e^{2\phi})=0\Longleftrightarrow \Phi=a_0e^{-2\phi}=a_0B^{-2}
\end{equation}
where $a_0$ is a complex constant.

Let us consider \eqref{eq:trace}. We notice that $R=-\ddot{H}$ and for any function $f(x)$
\begin{equation}
    \nabla_\mu \nabla^\mu f=(H\dot{f})'
\end{equation}
The equation \eqref{eq:trace} take a form
\begin{equation}
\left(H\frac{\dot{B}}{B}\right)'+H\left(\frac{\dot{B}}{B}\right)^2+\frac{1}{2}\ddot{H}+\Lambda+\frac{2|a_0|^2}{B^4}=0
\end{equation}
where we used formula for $\Phi$. Substituting formula for $B$ we obtain
\begin{equation}
    \ddot{H}+\frac{2}{x}\dot{H}+2\Lambda+\frac{2\alpha}{x^4}=0,\quad \alpha=2|a_0|^2
\end{equation}
The general solution of this equation has two free parameters $c_1,c_2$
\begin{equation}
    H=c_1+\frac{c_2}{x}-\frac{\Lambda}{6}x^2-\frac{\alpha}{x^2}=\frac{Q(x)}{x^2},\quad Q(x)=-\frac{\Lambda}{6}x^4+c_1x^2+c_2x-\alpha
\end{equation}
However, the conditions for the smooth metric excludes this case.

\begin{lm}
For every $H(x)$ as above the metric is not smooth.
\end{lm}

\begin{proof}
We can write
\begin{equation}
    Q(x)=f(x)(x-x_-)(x-x_+)
\end{equation}
where $f(x)$ is quadratic polynomial. Derivatives in $x_\pm$ simplify
\begin{equation}
    H'(x_+)=\frac{f(x_+)}{x_+^2}(x_+-x_-),\quad H'(x_-)=\frac{f(x_-)}{x_-^2}(x_--x_+).
\end{equation}
The condition for the absence of conical singularities is given by
\begin{equation}
    \frac{f(x_+)}{x_+^2}=\frac{f(x_-)}{x_-^2}=c_3<0
\end{equation}
and $L=-\frac{4\pi}{c_3}$.
The solution is
\begin{equation}
    f(x)=c_3x^2+c_4(x-x_-)(x-x_+)
\end{equation}
Comparing the coefficients at $x^3$ in the polynomial $Q(x)$
\begin{equation}
    0=-(c_3+c_4)(x_++x_-)-c_4(x_++x_-)=-c_3(x_++x_-)\Longrightarrow c_3+2c_4=0
\end{equation}
The constant term is given by
\begin{equation}
    -\frac{\alpha}{2}=c_4x_+^2x_-^2\Longrightarrow c_4<0
\end{equation}
This is a contradiction as $c_3<0$.

\end{proof}
The conclusion is, that 
$$\omega_\mu=0.$$
The conclusion of the proposition follows easily from the equations when $\omega_\mu$ is replaced by $0$. 
\end{proof}

\section{Summary}
We have derived the general static solution of the Einstein vacuum extremal horizon (EEH) equation (Definition \ref{extreq})  with $\lambda<0$. A detailed construction of the solution  was carried out  from admissible data that   (Theorem \ref{thm:generalstatic}).  In this way, we have completed the result obtained in \cite{Wylie2023}, also, we have presented  our own method for solving the equation.

We have found the general solution to the Einstein-Maxwell extremal horizon (EMEH) equations on $2$ dimensional compact manifold of non-positive Euler (Corollary \ref{cor:EMEHgenus}). Similar  topological censorship is known for  the EEH equations \cite{DKLS1, DKLS2} and for the EMEH equations with non-exact rotation potential (Definition \ref{etremmaxeq})  \cite{CKL-new}. For arbitrary  $2$ dimensional  manifold we have shown the rigidity of the  solutions to the EMEH equations such that the rotation potential is exact (Proposition \ref{prop:staticEMEH}). That implies the rigidity of every static solution to the EMEH equations considered on a manifold of the topology of $S^2$ and, combined with the very recent rigidity results on  non-static solutions \cite{CKL-new},  in the consequence the uniqueness of the Reissner-Nordstr{\"o}m-(Anti)-de-Sitter solution (Proposition \ref{prop:staticEMEH}).\footnote{The case of $\lambda>0$ and electromagnetic field was not explicitly considered in the literature (see Footnote 1 in \cite{CKL-new}).}  Those  results follow from our investigation of a more general family of equations we call generalized extremal horizon (GEH) equations  (Definition \ref{etremgeq}). Their solutions satisfy the topological censorship Proposition \ref{genus} and the static rigidity  Proposition \ref{staticrigidity}. Those conclusions are more general (in the case of the topological censorship) or complementary  (in the case of the rigidity)  to those of \cite{CKL-new}.

We have derived Master Identity (see Lemma \ref{lm:distiled}) for the EEH equation whose contraction with the Dunajski-Lucietti vector field provides their identity \cite{DL-rigidity}. It was that identity, that we applied later for the derivation of the static solutions of $\lambda<0$ discussed above, and for our  simplified  proof of the rigidity of the solutions to the EEH equation for arbitrary sign of $\lambda$ (Proposition \ref{pr:rigiditya}).

\bigskip

\noindent{\bf Acknowledgements:} Our thanks go to Maciek Dunajski and James Luccieti for providing us with useful information and turning our attention to reference \cite{Wylie2023}. 
JL was supported by the Polish National Science Centre grant No. 2021/43/B/ST2/02950. 

\appendix

\section{On the proof of the AMS theorem.}\label{sec:AMS-comments}

It is shown in \cite{AMS} that both $H$ and $H^\dagger$ have strictly positive principal eigenfunctions with eigenvalue $0$.
\begin{equation}
    H\psi_p=0,\quad H^\dagger \psi'_p=0
\end{equation}
where $\psi'_p=1$.
We will now explain why this eigenvalue is multiplicity free and simple and moreover, why every other eigenvalue has real part strictly bigger than $0$.
Let us remind that both $H$ and $H^\dagger$ have discrete spectra with generalized eigenspaces of finite order.

We consider an identity from \cite{AMS}, valid for arbitrary complex function $u$
\begin{equation}\label{eq:AMS-uu}
H^\dagger |u|^2=2\Re \left(\overline{u}H^\dagger(u) \right)-2\nabla_\mu \overline{u}\nabla^\mu u
\end{equation}
Suppose now that $u$ is an eigenfunction of $H^\dagger$ with eigenvalue $\mu$ such that $\Re\mu\leq 0$ then
\begin{equation}
    H^\dagger |u|^2=-\left(2\nabla_\mu \overline{u}\nabla^\mu u-2\Re\mu |u|^2\right).
\end{equation}
However, $\langle \psi_p,H^\dagger |u|^2\rangle=\langle H\psi_p, |u|^2\rangle=0$ thus
\begin{equation}
    0=\int \psi_p\left(2\nabla_\mu \overline{u}\nabla^\mu u-2\Re\mu |u|^2\right).
\end{equation}
As the integrand is non-negative it needs to vanish identically. In particular, $\nabla_\mu u=0$. However, this shows that $u$ is constant and $\mu=0$. We showed that every non-zero eigenvalue has real part bigger than $0$. Moreover, if $H^\dagger u=0$ then $u$ is constant.

Suppose now that $(H^\dagger)^2u=0$ but $H^\dagger u\not=0$ then
\begin{equation}
H^\dagger u=c\psi_p',\quad c\not=0
\end{equation}
However,
\begin{equation}
    0=\langle \psi_p, H^\dagger u\rangle=c\langle \psi_p,  \psi_p'\rangle
\end{equation}
Moreover, $\langle \psi_p,  \psi_p'\rangle\not=0$ as both functions are strictly positive, so $c=0$ a contradiction. We showed that generalized eigenspace for $H^\dagger$ is one dimensional.

The resolvents of $H$ and $H^\dagger$ are related
\begin{equation}
    (H-z){-1}=\overline{(H^\dagger-\overline{z})^{-1}}.
\end{equation}
Thus, in fact we proved also the desired properties for operator $H$.

Let $U_t$ be the unitary group commuting with $H$ and preserving space of non-negative functions. The eigenspace for $0$ is one dimensional thus $U_t\Psi=c_t\Psi$. From unitarity, $|c_t|=1$. However, as $U_t$ preserves the space of positive functions $c_t>0$. As result $c_t=1$.

\section{Derivation of Master Identity}\label{sec:MI-derivation}

We compute first $\square h_\mu$.

\begin{lm}
If EEH equation is satisfied then $\square h_\mu=(\nabla^\mu h_\mu)h_\nu+\Omega_{\mu\nu}h^\mu$.
\end{lm}

\begin{proof}
Let us notice that
\begin{equation}
\nabla_{(\mu}h_{\nu)}-\frac{1}{2}g_{\mu\nu}\nabla^\chi h_\chi=\frac{1}{2}h_\mu h_\nu-\frac{1}{4}g_{\mu\nu} h_\chi h^\chi-G_{\mu\nu}-\left(\frac{n}{2}-1\right)\lambda g_{\mu\nu}
\end{equation}
where $G_{\mu\nu}$ is the Einstein tensor. Taking divergence of this equation and multiplying by $2$
\begin{equation}
\square h_\nu=\nabla^\mu(h_\mu h_\nu)-\frac{1}{2}\nabla_\nu (h_\mu h^\mu)=(\nabla^\mu h_\mu)h_\nu +h_\mu \nabla^\mu h_\nu-h_\mu \nabla_\nu h^\mu.
\end{equation}
The result follows from $\Omega_{\mu\nu}=2\nabla_{[\mu} h_{\nu]}$.
\end{proof}

We notice that 
\begin{equation}
\square (\Gamma h_\mu)=(\Delta \Gamma)h_\mu+2(\nabla^\nu \Gamma) \nabla_\nu h_\mu+\Gamma \square h_\nu.
\end{equation}
We compute
\begin{equation}
(\nabla^\nu \Gamma) \nabla_\nu h_\mu=(\nabla^\nu \Gamma) \nabla_{(\nu} h_{\mu)}+(\nabla^\nu \Gamma) \nabla_{[\nu} h_{\mu]}=(\nabla^\nu \Gamma) \nabla_{(\nu} h_{\mu)}+\frac{1}{2}(\nabla^\nu \Gamma) \Omega_{\nu\mu}
\end{equation}
Moreover,
\begin{equation}
(\nabla^\nu \Gamma) \nabla_{(\nu} h_{\mu)}=(\nabla^\nu \Gamma) \left(\frac{1}{2} h_{\mu}h_\nu-R_{\mu\nu}\right)+\lambda \nabla_\mu \Gamma.
\end{equation}
Finally we compute
\begin{equation}
\Delta \nabla_\nu \Gamma=\nabla^\mu \nabla_\nu \nabla_\mu \Gamma= \nabla_\nu \nabla^\mu\nabla_\mu \Gamma+R_{\mu\phantom{\chi\mu}\nu}^{\phantom{\mu}\chi\mu}\nabla_\chi \Gamma=\nabla_\nu\Delta \Gamma+R^\chi_\nu \nabla_\chi \Gamma
\end{equation}
This means
\begin{equation}
\square \nabla_\nu \Gamma=\nabla_\nu\Delta \Gamma+2R^\chi_\nu \nabla_\chi \Gamma
\end{equation}
Together,
\begin{align}
\square \left(\nabla_\nu \Gamma+\Gamma h_\nu\right)=&\nabla_\nu \Delta\Gamma +2R^\chi_\nu \nabla_\chi \Gamma+(\nabla^\mu \Gamma) \left(h_{\mu}h_\nu-2R_{\mu\nu}\right)+\nonumber\\
&+\Delta \Gamma h_\mu+2\lambda \nabla_\nu \Gamma+(\nabla^\mu \Gamma) \Omega_{\mu\nu}+\Gamma(\nabla^\mu h_\mu)h_\nu+\Gamma h^\mu \Omega_{\mu\nu}.
\end{align}
The equality $\nabla^\mu (\Gamma h_\mu)=(\nabla^\mu \Gamma)h_\mu+\Gamma \nabla^\mu h_\mu$ allow us to write
\begin{equation}
\square \left(\nabla_\nu \Gamma+\Gamma h_\nu\right)=\nabla_\nu(\Delta \Gamma+2\lambda \Gamma)+\left(\Delta \Gamma+\nabla^\mu (\Gamma h_\mu)\right)h_\nu+\left(\nabla^\mu \Gamma+\Gamma h^\mu\right)\Omega_{\mu\nu}
\end{equation}
This is exactly the identity.

    \bibliographystyle{ieeetr}

\begin{thebibliography}{10}

\bibitem{ABF1}
A.~Ashtekar, C.~Beetle, and S.~Fairhurst, ``Isolated horizons: a generalization
  of black hole mechanics,'' {\em Class. Quant. Grav.}, vol.~16, no.~2,
  pp.~L1--L7, 1999.

\bibitem{ABL1}
A.~Ashtekar, C.~Beetle, and J.~Lewandowski, ``{Geometry of generic isolated
  horizons},'' {\em Class. Quant. Grav.}, vol.~19, pp.~1195--1225, 2002.

\bibitem{ABL2}
A.~Ashtekar, C.~Beetle, and J.~Lewandowski, ``{Mechanics of rotating isolated
  horizons},'' {\em Phys. Rev.}, vol.~D64, p.~044016, 2001.

\bibitem{AK}
A.~Ashtekar and B.~Krishnan, ``{Isolated and dynamical horizons and their
  applications},'' {\em Living Rev. Rel.}, vol.~7, p.~10, 2004.

\bibitem{LPextremal}
J.~Lewandowski and T.~Paw{\l}owski, ``{Extremal isolated horizons: A Local
  uniqueness theorem},'' {\em Class. Quant. Grav.}, vol.~20, pp.~587--606,
  2003.

\bibitem{LPhigher}
J.~Lewandowski and T.~Paw{\l}owski, ``{Quasi-local rotating black holes in
  higher dimension: Geometry},'' {\em Class. Quant. Grav.}, vol.~22,
  pp.~1573--1598, 2005.

\bibitem{LPJfol}
T.~Paw{\l}owski, J.~Lewandowski, and J.~Jezierski, ``{Space-times foliated by
  Killing horizons},'' {\em Class. Quant. Grav.}, vol.~21, pp.~1237--1252,
  2004.

\bibitem{Hajicek}
P.~{H{\'a}ji{\v{c}}ek}, ``{Three remarks on axisymmetric stationary
  horizons},'' {\em Communications in Mathematical Physics}, vol.~36,
  pp.~305--320, Dec. 1974.

\bibitem{IM}
V.~{Moncrief} and J.~{Isenberg}, ``{Symmetries of cosmological Cauchy
  horizons},'' {\em Communications in Mathematical Physics}, vol.~89,
  pp.~387--413, Sept. 1983.

\bibitem{Kunduri2007}
H.~K. Kunduri, J.~Lucietti, and H.~S. Reall, ``Near-horizon symmetries of
  extremal black holes,'' {\em Class.Quant.Grav.24:4169-4190,2007}, vol.~24,
  pp.~4169--4189, July 2007.

\bibitem{KunduriVac}
H.~K. Kunduri and J.~Lucietti, ``A classification of near-horizon geometries of
  extremal vacuum black holes,'' {\em J. Math. Phys.}, vol.~50, no.~8,
  p.~082502, 2009.

\bibitem{Kunduri:2008tk}
H.~K. Kunduri and J.~Lucietti, ``{Uniqueness of near-horizon geometries of
  rotating extremal AdS(4) black holes},'' {\em Class. Quant. Grav.}, vol.~26,
  p.~055019, 2009.

\bibitem{NHG}
H.~K. Kunduri and J.~Lucietti, ``{Classification of near-horizon geometries of
  extremal black holes},'' {\em Living Rev. Rel.}, vol.~16, p.~8, 2013.

\bibitem{Buk:2020ttx}
E.~Buk and J.~Lewandowski, ``{Axisymmetric, extremal horizons in the presence
  of a cosmological constant},'' {\em Phys. Rev. D}, vol.~103, no.~10,
  p.~104004, 2021.

\bibitem{Li}
C.~Li and J.~Lucietti, ``Uniqueness of extreme horizons in
  einstein{\textendash}yang{\textendash}mills theory,'' {\em Class. Quant.
  Grav.}, vol.~30, no.~9, p.~095017, 2013.

\bibitem{LiLucietti13}
C.~Li and J.~Lucietti, ``{Uniqueness of extreme horizons in Einstein-Yang-Mills
  theory},'' {\em Class. Quant. Grav.}, vol.~30, p.~095017, 2013.

\bibitem{DKLS2}
D.~Dobkowski-Ry\l{}ko, W.~Kami\'nski, J.~Lewandowski, and A.~Szereszewski,
  ``{The Near Horizon Geometry equation on compact 2-manifolds including the
  general solution for g \ensuremath{>} 0},'' {\em Phys. Lett. B}, vol.~785,
  pp.~381--385, 2018.

\bibitem{CRT}
P.~T. Chru{\'{s}}ciel, H.~S. Reall, and P.~Tod, ``On non-existence of static
  vacuum black holes with degenerate components of the event horizon,'' {\em
  Class. Quant. Grav.}, vol.~23, no.~2, pp.~549--554, 2005.

\bibitem{Bahuaud:2022iao}
E.~Bahuaud, S.~Gunasekaran, H.~K. Kunduri, and E.~Woolgar, ``{Static
  near-horizon geometries and rigidity of quasi-Einstein manifolds},'' {\em
  Lett. Math. Phys.}, vol.~112, no.~6, p.~116, 2022.

\bibitem{Wylie2023}
W.~{Wylie}, ``{Rigidity of compact static near-horizon geometries with negative
  cosmological constant},'' {\em Letters in Mathematical Physics}, vol.~113,
  p.~29, Apr. 2023.

\bibitem{DL-rigidity}
M.~Dunajski and J.~Lucietti, ``{Intrinsic rigidity of extremal horizons},''
  {\em arXiv preprint arXiv:2306.17512}, 6 2023.

\bibitem{Colling2024}
A.~Colling, M.~Dunajski, H.~Kunduri, and J.~Lucietti, ``New quasi-{E}instein
  metrics on a two-sphere,'' {\em arXiv preprint arXiv:2403.04117}, Mar. 2024.

\bibitem{CKL-new}
A.~Colling, D.~Katona, and J.~Lucietti, ``{Rigidity of the extremal Kerr-Newman
  horizon},'' {\em {arXiv preprint arXiv:2406.07128}}, 6 2024.

\bibitem{JK12}
J.~Jezierski and B.~Kami{\'n}ski, ``{Towards uniqueness of degenerate axially
  symmetric Killing horizon},'' {\em Gen. Rel. Grav.}, vol.~45, pp.~987--1004,
  2013.

\bibitem{CST}
P.~T. Chru\'sciel, S.~J. Szybka, and P.~Tod, ``{Towards a classification of
  vacuum near-horizons geometries},'' {\em Class. Quant. Grav.}, vol.~35,
  no.~1, p.~015002, 2018.

\bibitem{LSW}
J.~Lewandowski, A.~Szereszewski, and P.~Waluk, ``{Spacetimes foliated by
  nonexpanding and Killing horizons: Higher dimension},'' {\em Phys. Rev.},
  vol.~D94, no.~6, p.~064018, 2016.

\bibitem{AMS}
L.~Andersson, M.~Mars, and W.~Simon, ``Stability of marginally outer trapped
  surfaces and existence of marginally outer trapped tubes,'' {\em Advances in
  Theoretical and Mathematical Physics}, vol.~12, no.~4, pp.~853--888, 2008.

\bibitem{Bahuaud:2023wsi}
E.~Bahuaud, S.~Gunasekaran, H.~K. Kunduri, and E.~Woolgar, ``{Rigidity of
  quasi-Einstein metrics: the incompressible case},'' {\em Lett. Math. Phys.},
  vol.~114, no.~1, p.~8, 2024.

\bibitem{Cochran2024killing}
E.~Cochran, ``Killing fields on compact m-quasi-einstein manifolds,'' {\em
  arXiv preprint arXiv:2404.17090}, 2024.

\bibitem{Wald}
R.~M. Wald, {\em {General Relativity}}.
\newblock Chicago, USA: Chicago Univ. Pr., 1984.

\bibitem{HormanderPDEIII}
L.~H{\"o}rmander, {\em {The Analysis of Linear Partial Differential Operators
  III}}.
\newblock Springer Berlin, Heidelberg, 2007.

\bibitem{Kobayashi1995}
S.~Kobayashi, {\em {Transformation Groups in Differential Geometry}}.
\newblock Springer Berlin, Heidelberg, 1995.

\bibitem{Bochner46}
S.~Bochner, ``{Vector fields and Ricci curvature},'' {\em Bulletin of the
  American Mathematical Society}, vol.~52, pp.~776 -- 797, Sept. 1946.

\bibitem{DKLS1}
D.~Dobkowski-Ry\l{}ko, W.~Kami\'nski, J.~Lewandowski, and A.~Szereszewski,
  ``{The Petrov type D equation on genus \ensuremath{>}0 sections of isolated
  horizons},'' {\em Phys. Lett. B}, vol.~783, pp.~415--420, 2018.

\bibitem{Griffiths}
P.~Griffiths and J.~Harris, {\em {Principles of Algebraic Geometry}}.
\newblock John Wiley \& Sons, 1978.

\bibitem{Chrusciel:2005pa}
P.~T. Chru{\'s}ciel, H.~S. Reall, and P.~Tod, ``{On non-existence of static
  vacuum black holes with degenerate components of the event horizon},'' {\em
  Class. Quant. Grav.}, vol.~23, pp.~549--554, 2006.

\bibitem{LS}
J.~Lewandowski and A.~Szereszewski, ``Axial symmetry of {K}err spacetime
  without the rigidity theorem,'' {\em Phys. Rev. D}, vol.~97, p.~124067, 2018.

\end{thebibliography}

\end{document}